\documentclass[12pt,a4paper]{article}

\usepackage{amssymb,amsmath,amsfonts,eurosym,geometry,ulem,graphicx,caption,color,setspace,comment,caption,natbib,pdflscape,subfigure,array,hyperref,amsmath,amsthm,float,ascmac,bm,tikz,lipsum,sgamevar,natbib,url,enumerate}

\usepackage{tcolorbox}
\tcbuselibrary{theorems}

\usepackage{colortbl}

\usetikzlibrary{decorations.pathreplacing}
\usepackage[T1]{fontenc}
\usepackage{newpxmath}
\usepackage{newpxtext}
 \usepackage{hyperref}
 \usepackage{diagbox}
 \usepackage{threeparttable}
 \usepackage{comment}
\usetikzlibrary{patterns}

\usepackage[title]{appendix}

\usepackage[immediate]{silence}
\usepackage{sectsty}
\DeactivateWarningFilters[temp]

\allowdisplaybreaks

\theoremstyle{plain}\newtheorem{theorem}{Theorem}
\theoremstyle{plain}\newtheorem{proposition}{Proposition}
\theoremstyle{plain}\newtheorem{corollary}{Corollary}
\theoremstyle{plain}\newtheorem{lemma}{Lemma}
\theoremstyle{plain}

\theoremstyle{definition}
\newtheorem{definition}{Definition}

\newcommand{\supp}[1]{\text{supp}(#1)}

\newcommand{\argmax}{\mathop{\rm arg~max}\limits}

\usetikzlibrary{calc}
\usetikzlibrary{positioning}
\usetikzlibrary{matrix,fit}

\makeatletter
\let\@makefntextOring\@makefntext
\def\@makefntext#1{\@makefntextOring{\baselineskip=15pt#1}}
\makeatother

\definecolor{navyblue}{rgb}{0.0, 0.0, 0.5}
\definecolor{green}{rgb}{0.2, 0.5, 0.2}

\hypersetup{
    colorlinks=true,
    citecolor=navyblue,
    linkcolor=navyblue,
    urlcolor=navyblue,
}

\usepackage{multirow}

\usepackage{lscape}

\usepackage{booktabs}
\usepackage{multirow}
\usepackage{quoting}

\usepackage{etoolbox}
\patchcmd{\thanks}{#1}{\protect\doublespacing}{}{}

\DeclareFontShape{OT1}{cmr}{m}{n}{<->cmr10}{}
\usepackage{fix-cm}

\usepackage{cancel}

\usepackage{subfiles} 
\usepackage{hyperref}  

\usepackage{titlesec}

\usepackage{xr}

\title{Value of History in Social Learning: Applications to Markets for History\thanks{\protect \onehalfspacing 
Sato acknowledges the financial support from the JSPS KAKENHI Grant 24KJ0100.
Shimizu acknowledges the financial support from the JSPS KAKENHI Grant 23KJ0667. 
All remaining errors are our own.
}}
\author{
\Large Hiroto Sato\thanks{\protect \onehalfspacing 
Department of Economics, Nagoya University, Furo-cho, Chikusa-ku, Nagoya, 464-8601, Japan. 
Email: \url{sato.hiroto.s9@f.mail.nagoya-u.ac.jp}.
}
\and
\Large Konan Shimizu\thanks{\protect \onehalfspacing 
Graduate School of Economics, The University of Tokyo, 7-3-1, Hongo, Bunkyo-ku, Tokyo, 113-0033, Japan. 
Email: \url{shimizu-konan@g.ecc.u-tokyo.ac.jp}.}
} 
\date{\today}

\begin{document}

\begin{titlepage}
\maketitle

\begin{abstract}
    In social learning environments, agents acquire information from both private signals and the observed actions of predecessors, referred to as history. 
    We define the {\it value of history} as the gain in expected payoff from accessing both the private signal and history, compared to relying on the signal alone. 
    We first characterize the information structures that maximize this value, showing that it is highest under a mixture of full information and no information. 

    We then apply these insights to a model of {\it markets for history}, where a monopolistic data seller collects and sells access to history. 
    In equilibrium, the seller's dynamic pricing becomes the value of history for each agent.
    This gives the seller incentives to increase the value of history by designing the information structure.
    The seller optimal information discloses less information than the socially optimal level.
\end{abstract}


\setcounter{page}{0}
\thispagestyle{empty}
\end{titlepage}

\newcolumntype{C}[1]{>{\centering\arraybackslash}p{#1}}

\setlength{\abovedisplayskip}{5pt}
\setlength{\belowdisplayskip}{5pt}

\section{Introduction} \label{sec:introduction}
In many economic and social environments, individuals make decisions not only based on their own private information, but also by observing the actions of others—a phenomenon known as social learning. 
The sequence of observed past actions is referred to as {\it history}. 
This informational externality gives rise to an instrumental value of history, much like the value of information. 
Understanding the the value of history can help clarify how its value depending on the information structure and how platforms or institutions should manage or monetize it.

To illustrate this point, the value of history is becoming increasingly significant in the modern, data-driven economy.
Maintaining and storing historical data can be costly, particularly for digital platforms that support long-term information aggregation or recommendation systems.
Quantifying the value of history is therefore essential for justifying these costs and evaluating the return on investment in record-keeping infrastructure.
Moreover, history is not merely a passive record—it can be actively transacted, giving rise to what we term {\it markets for history},
historical data is treated as a commodity, bought and sold among participants, and thus the value of history directly shapes its price.
Examples include financial markets, where platforms like Bloomberg sell data on the past actions of investors, and credit reporting, where agencies such as Equifax and FICO aggregate and monetize lending and repayment histories.

Our baseline model builds on the classical framework of sequential social learning \citep{banerjee1992simple,bikhchandani1992theory,smith2000pathological}. 
In this model, there exists a binary, payoff-relevant state of the world. 
A sequence of homogeneous agents make binary decisions one after another. 
Each agent observes the actions of all predecessors (history) and also receives a private signal independently drawn from an identical information structure. 
We define the {\it value of history} for agent $i$ as the difference in expected payoff of agent $i$ between the case where both the private signal and history are observed and the case where only the private signal is used. 
Aggregating across agents, we define the {\it social value of history} as the discounted sum of individual values of history over all agents.

To benchmark our analysis, we first observe that under full information (i.e., perfectly informative private signals) or no information (i.e., uninformative signals), the value of history for each agent is zero. 
Thus, the value of history arises when agents face the intermediate precision of information.
This observation raises a natural question: {\it Which information structure maximizes the value of history?}
or {\it What is the upper bound of the value of history without observing the underlying information structure?}
Our result shows that, for any information structure, there exists an information structure consists of the mixture of full and no information such that it achieves the higher value of history for all agents.
Technically, our result relies on the analysis of \citet{sato2025value}.\footnote{They analyze the comparison of information structures in social learning model.}

This result highlights a subtle interaction between private and social sources of information and has two main implications.
The first is metaphorical. 
Specifically, our results provide a prediction set for the value and the social value of history, even without knowledge of the underlying information structures.
This interpretation can inform the regulation of information environments in settings where social learning is present.
For example, when a platform manages a system that records past behavior, the existence of such a system is justified under some information structure if and only if its cost is below the upper bound of the social value of history.

The second interpretation is more literal and relates to the design and management of the information structure, which naturally arises in the context of markets for history, as mentioned earlier.
In this extended model, a monopolistic data seller collects and records the history taken by agents. 
We assume that the seller can collect past actions without any cost by the agents who purchased history, abstracting away incentive and enforcement issues.
The seller then sets a dynamic price and offers access to the history to future agents.\footnote{This assumption is just for simplification. 
We can extend the model in which the seller can only change the price at some specific periods.
We discuss this point in Section \ref{sec: sticky price}.} 
Before the decision making, agents may choose to purchase the history and after then they observe the private signals, which shapes the posterior beliefs.\footnote{The timing of observations captures the following situation: Investors repeatedly determine whether they invest to a particular opportunity. 
Then, they usually contract with Bloomberg in front of each inspection related with investment decisions. 
}
Although this setup abstracts from some complexities (e.g., the joint sale of private information and history), it captures the essential strategic and informational dynamics of markets for history. 
It allows us to analyze how the value of history shapes market demand, how prices are formed, and how welfare and efficiency are affected by the emergence of such data markets.

In this set-up, buyers purchase the history if and only the value of purchasing the history, which is the expected payoffs by observing both the history and private signal minus price, is higher than the payoffs by only observing the private signal.
Thus, the seller can increase the price of each buyers up to the value of history, and hence at the equilibrium, the seller makes take-it-or-leave-it offers at this level.
By this observation, the seller surplus at the equilibrium becomes the social value of history.
Moreover, the dynamic pricing is efficient in this market as it does not produce agents do not purchase the history (where the history of such agents is not recorded). 

By this observation, the social surplus and the seller surplus are determined by the precision of the underlying information structure.
We then apply Theorem \ref{thm:ternary_signal} to clarify the role of information in the markets for history.
Proposition \ref{prop: buyer_seller_socially_optimal_info} shows that the buyer, the seller, and the social welfare (which is defined by the weighted sum of them) are maximized by some mixture of full and no information.
This revelation principal gives us the closed form expression of these surpluses.
It shows that the buyer surplus is maximized by full information.
Moreover, the seller surplus is maximized by the intermediate information, which closes to no information as the seller becomes enough patient.
Thus, the social surplus is maximized at the level between the buyer-optimal and the seller-optimal levels.

\subsection{Related Literature}
The pioneering works by \citet{blackwell1951comparison, blackwell1953equivalent} investigate the ordinal ranking of information structures across the universal domain of decision problems.\footnote{
In contrast, some existing studies compare experiments within particular types of decision problems or restricted classes of experiments \citep{lehmann1988comparing, persico2000information, athey2018value, ben2024new}.
See also \citet{de2017instrumental} for an analysis of the cardinal value of information, and \citet{frankel2019quantifying} for an ex post approach to quantifying information and uncertainty.
}
Given a fixed decision problem, the value of information is easily quantified by the instrumental value of the information structure.
In contrast to this extensive literature, we focus on the value of history, defined as the marginal value of the additional information obtained through observing history.

Our formulation relates to the literature on social learning, pioneered by \citet{banerjee1992simple} and \citet{bikhchandani1992theory}, with a central issue, raised by \citet{smith2000pathological}, being whether asymptotic learning occurs.\footnote{Contrary to this literature, \citet{liang2020complementary} studies social learning with correlated information structures. Other important questions include convergence rates, as explored by \citet{hann2018speed} and \citet{rosenberg2019efficiency}.} 
Research on social learning with limited or noisy observation has been widely studied \citep{ccelen2004observational, acemoglu2011bayesian, lobel2015information, le2017information, arieli2019multidimensional, arieli2021general, kartik2024beyond,xu2023social},\footnote{See also \citet{banerjee2004word}, \citet{gale2003bayesian}, \citet{callander2009wisdom}, and \citet{smith2013rational} for studies on observational learning where agents observe only summary statistics of past actions.} alongside work on costly observation, which endogenously limits agents' ability to observe \citep{kultti2006herding, kultti2007herding, song2016social}.\footnote{See \citet{hendricks2012observational}, \citet{mueller2016social}, and \citet{ali2018herding} for studies on costly (direct/private) information.} 
We focus on the value of history and in our application for the markets for history, agents must pay a price, strategically set by the seller, to access the entire history of actions.
Thus, our model belongs to the literature of learning with limited or costly observation of history, but the primal focus is the interplay between the value of history and the information structure rather than whether the asymptotic learning occurs.

Recent studies on social learning focus on information design, balancing the benefits and costs of herd behavior (\citet{arieli2023herd,lorecchio2022persuading,parakhonyak2023information,le2017information,arielipositive}).\footnote{See also the mechanism design approach by \citet{sgroi2002optimizing} and \citet{smith2021informational}.} 
These papers usually focus on the objective about the probability of taking particular action or agents' expected payoffs in the limit.
Technically, we solve the information design problem to maximize the value of history and apply this result into the maximization of the social welfare in the market of history. 

Our focus of applications departs from the growing literature on markets for information \citep{admati1986monopolistic, admati1988selling, admati1990direct, esHo2007price, bergemann2015selling, bergemann2018design, mekonnen2023persuaded}.\footnote{See the unified survey by \citet{bergemann2019markets} for an overview of this literature.} 
Unlike these works, which emphasize selling direct information, we focus on markets for history, where the seller prices the history of past actions. 
The notable feature of the market for history is that the seller's pricing affects the accumulation of information in the society.
In this sense, our focus is the market transacting data with the process of information aggregation.


\section{Model} \label{sec:model}
There are binary state $\Omega=\{L,H\}$ with uniform prior.
There is an infinite sequence of agents indexed by $i=1,2,\dots,$ arriving sequentially.
The periods are discrete ($t=0,1,\dots$), and each agent $i$ takes an action at period $i$ from a binary action set $A=\{0,1\}$.
A common payoff function $u: A\times \Omega\to \mathbb{R}$ determines each agent's payoff. 
The payoff of agent $i$ depends solely on their own action and the state, independent of actions taken by other agents.
As in the literature, we assume that payoff function $u:A\times \Omega\to \mathbb{R}$ is given by\footnote{The analysis is same if we alternatively consider the setting in which $u(a,\omega)=1$ if $a=\omega$ and $u(a,\omega)=0$ otherwise.
} 
\[
u(a,\omega)=
\begin{cases}
1/2 & \text{ if } a=1 \text{ and } \omega=H, \\
-1/2 & \text{ if } a=1 \text{ and } \omega=L, \\
0 & \text{ otherwise }.
\end{cases}
\]

The timing of this game is as follows:
At period $0$, nature first determines the true state, which remains unchanged throughout the game. 
In each period $i$, agent $i$ first observes the entire history, which consists of the actions of all preceding agents ($1,2,\dots, i-1$). 
Additionally, agent $i$ receives a private signal $s\in S$, drawn (conditionally) independently from an identical information structure $\pi: \Omega\to\Delta(S)$.
For simplicity, we assume that $S$ is finite.
Following these observations, agent $i$ selects an action from the action set $A$.
Let $\sigma_{i}: A^{i-1}\times S\to \Delta(A)$ be a strategy of agent $i$, which is a mapping from history and private signal to the distribution over actions.
We denote $\bm{\sigma}$ as a strategy profile.

Let $V_{i}(\pi,\bm{\sigma})$ be the ex-ante expected payoff of agent $i$ at $\bm{\sigma}$ under $\pi$, that is, 
\[
    V_i(\pi,\bm{\sigma})=  \mathbb{E}_\omega\left[\sum_{(a_1,\cdots,a_i)\in A^i} \sum_{(s_1,\dots,s_i)\in S^{i}}\prod_{k=1}^{i}\pi(s_k|\omega)\sigma_k(a_k|a_1,\dots,a_{k-1},s_k)u(a,\omega)\right]. 
\]
Moreover, for simplicity, we denote $V(\pi)=V_{1}(\pi,\bm{\sigma})$. 

We say that the strategy profile $\bm{\sigma}^*$ is a Bayes-Nash equilibrium (hereafter referred to as an equilibrium) under $\pi$ if
\begin{align*}
    V_i(\pi,\bm{\sigma}^*)\geq V_i(\pi,(\sigma_i,\bm{\sigma_{-i}^*}))
\end{align*}
for all $\sigma_i$ and $i$. Define $BNE(\pi)$ as the set of all equilibria under $\pi$.
\section{Value of History} \label{sec:main}
Our focus is the instrumental value of history as defined in the following.
\begin{definition}
    The {\it value of history for agent $i$} is 
    \[
    \mathcal{V}_{i}(\pi)=\max_{\bm{\sigma}\in B^i(\pi)}V_{i}(\pi,\bm{\sigma})-V(\pi)
    \]
    where $B^1(\pi)=BNE(\pi)$ and $B^i(\pi)=\argmax_{\bm{\sigma}\in B^{i-1}(\pi)}  V_{i-1}(\pi,\bm{\sigma})$ for $i=2,3,\dots$.

    Moreover, the {\it social value of history} is 
    \[
    \mathcal{V}(\pi)=(1-\delta)\cdot \sum_{i=1}^{\infty}\delta^{i-1}\mathcal{V}_{i}(\pi)
    \]
    where $\delta\in(0,1)$ is the discount factor.
\end{definition}
Our primal focus is the value and the social value of history, and thus we introduce the equivalent class of information structures inducing same values.
\begin{definition}
    We say $\pi$ and $\pi'$ are {\it equivalent information structure} if $V(\pi)=V(\pi')$ and $  \mathcal{V}_i(\pi)=\mathcal{V}_i(\pi') $
    for all $i$. Write $\pi \sim \pi'$ when $\pi$ and $\pi'$ are equivalent.
\end{definition}
When the information structure is fixed, there may be multiple equilibria, in which case each agent is assumed to choose the strategy that maximizes its own payoff with the highest social value of history.\footnote{The equilibrium selection rules do not matter for our main results.}

Under these definitions, the value of history for agent $i$ and the social value of history are zero when $\pi$ is either full information or no information.
Thus, when each agent receives the signal from the extreme information structures, there is no instrumental value in the history, and hence the value of history arises only under the intermediate information structures.

When is the value of history maximized?
What value of history may be possible under some information structure?
Irrespective of the determinant of the information structure, deriving the upper bound of the value of history may provide insights into the robust prediction or the information design for example.
Specifically, if the cost of recording the history exceeds the maximized value of history, then it's system cannot be justified by any information structure.
We discuss one potential application related with the information design in the markets for history in Section \ref{sec:markets_for_history}.

\begin{theorem}\label{thm:ternary_signal} 
Take any information structure $\pi$ and equilibrium $\bm{\sigma}$ under $\pi$.
    \begin{enumerate}
        \item There exists an information structure $\pi^{*}$ and equilibrium $\bm{\sigma}^{*}$ under $\pi^{*}$ such that (i) $\supp{\mu^{*}}=\{0,1/2,1\}$ and (ii) $\mathcal{V}_{i}(\pi^{*})\geq \mathcal{V}_{i}(\pi)$ for all $i$.
        \item  If the value of history for agent $i$ ($i\geq 2$) is maximized at $\pi$, then there exists $\pi'$ such that $\pi \sim \pi'$ and $\supp{\mu'}\subset\{0,\frac{1}{2},1\}$.
    \end{enumerate}
\end{theorem}
Intuitively, the proof goes as follows:
First, as the prior and the cutoff of decision problem are coincide with $1/2$, splitting the interior private beliefs into $1/2$ and $1$ or $0$ and $1/2$ does not change the expected payoff from only observing private signals.
Then, utilizing the results of \citet{sato2025value}, we show that it weakly increases the expected payoff for all agents $i\geq 2$. In particular, it can be shown that the payoff strictly increases in most cases. This is because, when the original information structure includes both a signal that makes the state likely to be 
$H$ (but is not a conclusive signal) and a signal that makes the state likely to be 
$L$ (again, not conclusive), observing both actions that indicate these signals results in a strictly lower expected payoff compared to the case where agents receive signals from the split information structure. See Lemma 1 for details.
Conversely, when such signals do not exist, it can be shown that the expected payoffs for all agents are equal between the two information structures before and after the split. That is, the two structures are equivalent (Lemma 2). Therefore, we conclude that the only signals that can maximize the value of history are ternary signals or those equivalent to ternary signals.

Theorem \ref{thm:ternary_signal} indicates that the optimal information structure maximizing the value of history for agent $i$ can be found within the class of information structures satisfying $\supp{\mu}=\{0,1/2,1\}$.
Let $\varepsilon$ be the probability of disclosing the uninformative signal, that is, $\varepsilon=\pi(\mu=1/2|L)=\pi(\mu=1/2|H)$.
Denote $\pi(\varepsilon)$ be the information structure such that $\supp{\mu}=\{0,1/2,1\}$ and $\varepsilon=\pi(\mu=1/2|L)=\pi(\mu=1/2|H)$.
Then, we can derive the next proposition.
\begin{proposition}\label{prop_value_of_history_maximized}
    The value of history for agent $i$ is maximized at $\varepsilon_{i}^{*}=(1/i)^{i-1}$.
    Moreover, the value of history is maximized at $\varepsilon_{S}^{*}=(1-\sqrt{1-\delta})/\delta$, which is increasing in $\delta$.
\end{proposition}
\begin{proof}
Under $\pi(p)$, the expected payoff of agent $i$ is given by
\begin{align*}
    \mathcal{V}_{i}(\pi(\varepsilon))&=\frac{1}{4}\cdot (1-\varepsilon^{i})  - \frac{1}{4}\cdot (1-\varepsilon) \\
    & = \frac{1}{4}\cdot[\varepsilon-\varepsilon^{i}].
\end{align*}
Also, 
\begin{align*}
    \mathcal{V}(\pi(\varepsilon))=\frac{1}{4}\cdot (1-\delta)\sum_{i=1}^{\infty}\delta^{i-1}(\varepsilon-\varepsilon^{i}).
\end{align*}
Both the value of history for agent $i$ and society are concave in $\varepsilon$.\footnote{
By a simple calculation, we have
\begin{align*}
    \sum_{i=1}^{\infty}\delta^{i-1}(\varepsilon-\varepsilon^{i}) 
    & = \varepsilon\cdot \sum_{i=1}^{\infty}\delta^{i-1}(1-\varepsilon^{i-1}) \\
    & = \varepsilon\cdot \left[\frac{1}{1-\delta}-\frac{1}{1-\delta \varepsilon}\right] \\
    & = \frac{\delta \varepsilon(1-\varepsilon)}{(1-\delta)(1-\delta \varepsilon)}.
\end{align*}
Let $f(\varepsilon)= \frac{\varepsilon(1-\varepsilon)}{1-\delta \varepsilon}$.
Then, $f'(\varepsilon)=\frac{1}{(1-\delta \varepsilon)^{2}}\cdot [(1-2\varepsilon+\delta \varepsilon^{2}]$ and $f''(\varepsilon)=\frac{1}{(1-\delta \varepsilon)^{3}}\cdot (-2) \cdot (1-\delta-2\delta \varepsilon)$.
Thus, $f''(\varepsilon)\leq 0$ for all $\varepsilon\in (0,1)$, and thus $f$ is concave in $\varepsilon\in (0,1)$.
Moreover, as $f'(0)=1$ and $f'(1)=\frac{-1+\delta}{(1-\delta)^{2}}<0$, the interior solution is guaranteed.
}
Thus, there is a unique maximizer $\varepsilon_{i}^{*}$ and $\varepsilon_{S}^{*}$, respectively.
The simple calculation shows that $\varepsilon_{i}^{*}=(1/i)^{i-1}$.
Moreover, $\varepsilon_{S}^{*}$ is the solution of $1-2\varepsilon+\delta \varepsilon^{2}=0$ satisfying $\varepsilon_{S}^{*}\in(0,1)$.
Thus, we obtain $\varepsilon_{S}^{*}=(1-\sqrt{1-\delta})/\delta$.
Increasing $\delta$ implies higher $\varepsilon_{S}^{*}$, and thus noisier information structure. 
\end{proof}

\begin{corollary}
    The maximized value of history is weakly increasing in $\delta$. 
    Moreover, the maximized social value of history is $\frac{1}{4}\cdot \frac{(1-\sqrt{1-\delta})^{2}}{\delta}$, which approaches to $\frac{1}{4}$ as $\delta\to 1$.
\end{corollary}
Thus, the possible range of the values of history is larger when the society puts higher weight on the future.
When the platform or policy maker evaluates the system accumulating history, we can detect that higher costs than $\frac{1}{4}\cdot \frac{(1-\sqrt{1-\delta})^{2}}{\delta}$ indicates that the return is negative from that system without knowing the underlying information structures.

\section{Application: Markets for History} \label{sec:markets_for_history}
\subsection{Set-Up}
Throughout this section, the model consists of a risk-neutral seller and countably infinite buyers indexed by $i=1,2,\dots$. 
The basic model is same, but buyers can observe the entire history only when they pay the price of history for the seller.
The price of history is determined by the monopoly data seller.

The game proceeds as follows. 
At the beginning of each period $T=1,2,\cdots$, the seller posts a price for histories. Then buyer $i=T$ chooses whether to purchase a history of actions taken by previously moved buyers. 
Buyer $i$ later observes a private signal, which is drawn independently from the identical information structure $\pi:\Omega\rightarrow\Delta(S)$. 
Then, at the end of the period, the buyer takes action. 

Dynamic pricing is assumed here. That is, the seller can offer different prices to the buyer each period. Also, it is assumed that when a buyer purchases history, they can only observe the actions of agents who have previously bought history. That is, if there was a buyer in the past who did not buy the history, he cannot observe the actions of that buyer.
Additionally, each buyer can observe the accumulated amount of history, i.e., the number of agents who have purchased in the past, prior to making a purchase.

This model considers the following situations. 
There are a platform and agents, and each agent decides whether to make an investment. 
The platform charges the agent a registration fee, and the agent can receive a private signal without registering with the platform and decide whether or not to invest based on that alone.
However, if she registers with the platform, she can additionally observe the behavior of people who have registered with the platform in the past.

For a given price $p$, buyer $i$ purchases a history if $p\leq \mathcal{V}_{N(i)+1}(\pi)$, where $N(i)$ is the number of agents who purchased histories from $1$ to $i-1$.
Thus, the value of history determines their willingness to pay.
Note that in equilibrium, $N(i)=i-1$ always holds because the seller should let all buyers purchase history.

As the seller's revenue increases with price, in equilibrium, the seller offers a price $p_i(\pi)=\mathcal{V}_i(\pi)$ to the buyer $i$ and all buyers will buy history. Hence, the seller's surplus $\mathcal{V}_S(\pi)$ is calculated by
\begin{align}
    \mathcal{V}_{S}(\pi)
    = (1-\delta)\cdot \sum_{i=1}^{\infty}\delta^{i-1}p_i(\pi)
    = \mathcal{V}(\pi).
\end{align}
Let $\mathcal{V}_{B}(\pi)$ be the discounted average of the buyer's surplus. Then, 
\begin{align}
    \mathcal{V}_{B}(\pi) = (1-\delta)\cdot \sum_{i=1}^{\infty} \delta^{i-1}V(\pi)=V(\pi).
\end{align} 
Define the social surplus $\mathcal{V}^*(\pi)$ as a weighted average of the seller's surplus and the buyer's surplus. Formally, 
\begin{align*}
    \mathcal{V}^*(\pi)=\alpha \mathcal{V}_B(\pi)+(1-\alpha)\mathcal{V}_S(\pi)
\end{align*}
for some $\alpha \in (0,1)$.
\subsection{Optimal Information Structures}
    Using Theorem \ref{thm:ternary_signal}, we can characterize the information structure that maximizes the seller's surplus, the buyer's surplus, and the social surplus, respectively.
    \begin{proposition} \label{prop: buyer_seller_socially_optimal_info}
$\mathcal{V}_S(\pi)$, $\mathcal{V}_B(\pi)$, and $\mathcal{V}^*(\pi)$ are maximized when $\pi$ is equivalent to the ternary signal $\pi'$ with $\supp{\mu'}\subset\{0,\frac{1}{2},1\}$. 
Let $\varepsilon^*_S$, $\varepsilon^*_B$, and $\varepsilon^*$ be the values of the probability of receiving an uninformative signal that maximize $\mathcal{V}_S(\pi)$, $\mathcal{V}_B(\pi)$, and $\mathcal{V}^*(\pi)$, respectively. 
Then, $\varepsilon^*_S=\frac{1-\sqrt{1-\delta}}{\delta}$, $\varepsilon^*_B=0$, and $\varepsilon^*=0$ if $\delta\leq \frac{\alpha}{1-\alpha}$ and $\varepsilon^*=\frac{1-\sqrt{\frac{(1-\alpha)(1-\delta)}{1-2\alpha}}}{\delta}$ if $\delta\geq \frac{\alpha}{1-\alpha}$.
\end{proposition}

\begin{proof}
    Since $\mathcal{V}_B(\pi)=V(\pi)$ and $V(\pi)$ is maximized at full information, $\mathcal{V}_B(\pi)=V(\pi)$ is maximized when $\supp{\mu}=\{0,1\}$, or $\varepsilon^*_B=0$.
    
    Note that $\mathcal{V}_S(\pi)=\mathcal{V}(\pi)$. 
    Thus, Theorem~\ref{thm:ternary_signal} yields that $\mathcal{V}_S(\pi)$ is maximized when $\pi$ is equivalent to the ternary signal $\pi'$ with $\supp{\mu'}\subset\{0,\frac{1}{2},1\}$. 
    Also, by Proposition~\ref{prop_value_of_history_maximized}, we have $\varepsilon_S^*=\frac{1-\sqrt{1-\delta}}{\delta}$.
    
    Consider $ \mathcal{V}^*(\pi)$. Take any $\pi$ and consider $\pi^*$, which is the information structure derived from the split mentioned in Lemma \ref{lem:split}. 
    Then, we have $V(\pi)=V(\pi^*)$ and $\mathcal{V}(\pi^*)\geq \mathcal{V}(\pi)$ by Theorem~\ref{thm:ternary_signal}. 
    Also, $\mathcal{V}(\pi^*)= \mathcal{V}(\pi)$ holds if and only if $\pi\sim \pi^*$. 
    Since $\mathcal{V}^*(\pi)=\alpha \mathcal{V}_B(\pi)+(1-\alpha)\mathcal{V}_S(\pi)=\alpha V(\pi)+(1-\alpha)\mathcal{V}(\pi)$ holds, $\mathcal{V}^*(\pi)$ is maximized when there exists $\pi^*$ such that $\pi\sim \pi^*$ and $\supp{\mu^*}\subset \{0,\frac{1}{2},1\}$. 
    Now, it follows that 
    \begin{align*}
        \mathcal{V}(\pi(\varepsilon))=&\frac{1}{4}(1-\delta)\sum_{i=1}^\infty \delta^{i-1}(\varepsilon-\varepsilon^i)\\
       =&\frac{\delta \varepsilon(1-\varepsilon)}{4(1-\delta \varepsilon)}\\
        V(\pi(\varepsilon))=&\frac{1-\varepsilon}{4}.
    \end{align*}
    Hence, 
    \begin{align*}
        4\mathcal{V}^*(\pi(\varepsilon))&=\alpha(1-\varepsilon)+(1-\alpha)\frac{\delta \varepsilon(1-\varepsilon)}{1-\delta \varepsilon}\\
       &= \alpha-\alpha\varepsilon+(1-\alpha)\frac{(1-\delta\varepsilon)\varepsilon+\delta\varepsilon-\varepsilon}{1-\delta \varepsilon}\\
        &=\alpha+(1-2\alpha)\varepsilon+\frac{(1-\alpha)\varepsilon(\delta-1)}{1-\delta\varepsilon}.
    \end{align*}
    It follows that
     \begin{align*}
        4\frac{\partial}{\partial\varepsilon}\mathcal{V}^*(\pi(\varepsilon))&=1-2\alpha+\frac{(\delta-1)(1-\delta\varepsilon)+\delta\varepsilon(\delta-1)}{(1-\delta\varepsilon)^2}\\
        &=1-2\alpha-\frac{(1-\alpha)(1-\delta)}{(1-\delta\varepsilon)^2}\\
        4\frac{\partial^2}{\partial\varepsilon^2}\mathcal{V}^*(\pi(\varepsilon))&=-2\delta(1-\alpha)(1-\delta)(1-\delta\varepsilon)^{-3}<0.
    \end{align*}
    Thus, $\frac{\partial}{\partial\varepsilon}\mathcal{V}^*(\pi(\varepsilon))$ is strictly decreasing with respect to $\varepsilon$. 
    Note that $\frac{\partial}{\partial\varepsilon}\mathcal{V}^*(\pi(1))=1-2\alpha-\frac{1-\alpha}{1-\delta}=\frac{2\alpha\delta-\delta-\alpha}{1-\delta}< 0$ always holds and $\frac{\partial}{\partial\varepsilon}\mathcal{V}^*(\pi(0))=\delta-\alpha-\alpha\delta\leq 0$ if and only if $\delta \leq\frac{\alpha}{1-\alpha}$. 
    Therefore, $\mathcal{V}^*(\pi(\varepsilon))$ is maximized at $\varepsilon=0$ if $\delta \leq\frac{\alpha}{1-\alpha}$ and at $\varepsilon=\varepsilon^*$ if $\delta \geq\frac{\alpha}{1-\alpha}$, where $\varepsilon^*$ satisfies $\frac{\partial}{\partial\varepsilon}\mathcal{V}^*(\pi(\varepsilon^*))=0$, i.e., $\varepsilon^*=\frac{1}{\delta}(1-\sqrt{\frac{(1-\alpha)(1-\delta)}{1-2\alpha)}})$.
 \end{proof}
For buyers, full information is optimal because they can obtain the maximum value without having to purchase a history if they know the state only from their private signals.

However, full information is not a desirable structure for the seller. 
The seller can sell the history only at a price no higher than the value of history, but as we have seen, under full information, the value of history becomes zero. 
As a result, the seller's surplus also becomes zero. 
The optimal information structure for the seller is a mixture signal combining full information and no information, where the probability of receiving an uninformative signal in each state is $\frac{1 - \sqrt{1 - \delta}}{\delta}$. 
This probability increases monotonically with $\delta$ and converges to 1 as $\delta \to 1$. 
Therefore, if the seller is sufficiently patient, a signal that is nearly completely uninformative becomes optimal.

If $\alpha$ is less than $\frac{1}{2}$ and $\delta$ is sufficiently large, then the socially optimal signal also becomes a mixture of full information and no information.
\subsection{Sticky Prices Model} \label{sec: sticky price}
Here, we consider the case in which the seller cannot change the price in every period. 
In particular, we examine the optimal information structure when the seller can adjust the price only at every $ t$-period interval. 
In other words, the seller can set the price only at periods $1, t+1, 2t+1, 3t+1,\dots$. 
All other settings are the same as before. 
Note that in equilibrium, the seller offers the history at price $p = \mathcal{V}_{kt+1}(\pi)$ to buyers $kt+1$ to $(k+1)t$. 
This is because $\mathcal{V}_i(\pi)$ is monotonically increasing in $i$, and the history of buyers who do not purchase the history is not accumulated.

In this case, the seller's surplus $\mathcal{V}_S(\pi|t)$ is calculated by
\begin{align*}
    \mathcal{V}_{S}(\pi|t)
    = (1-\delta)\cdot \sum_{i=1}^{\infty}\delta^{i-1}p_i(\pi|t)
    =(1-\delta^t) \sum_{k=0}^\infty \delta^{kt}\mathcal{V}_{kt+1}(\pi).
\end{align*}
Similarly, the buyer's surplus  $\mathcal{V}_B(\pi|t)$ and social surplus  $\mathcal{V}^*(\pi|t)$ are
\begin{align*}
   & \mathcal{V}_{B}(\pi|t) = (1-\delta)\cdot \sum_{i=1}^{\infty} \delta^{i-1}[V(\pi)+\mathcal{V}_i(\pi)-p_i(\pi|t)].\\
   & \mathcal{V}^*(\pi|t)=\alpha \mathcal{V}_B(\pi|t)+(1-\alpha)\mathcal{V}_S(\pi|t)
\end{align*}
for some $\alpha \in (0,1)$.

Note that $V(\pi)+\mathcal{V}_i(\pi)-p_i(\pi|t)\leq \max_{\bm{\sigma}\in B^i(\pi)} V_i(\pi,\bm{\sigma}) \leq \frac{1}{4}$, where the last inequality holds strictly if $\pi$ is not full information. 
In addition, it can be calculated $V(\pi)+\mathcal{V}_i(\pi)-p_i(\pi|t)=\frac{1}{4}$ when $\pi$ is full information. 
Thus, the information structure that maximizes buyer surplus in this case is still full information.
Furthermore, since the seller's surplus is monotonically increasing with respect to $\mathcal{V}_i(\pi)$ (for $i = 1, t+1, 2t+1, \dots$), it follows from Theorem~\ref{thm:ternary_signal} that the information structure which maximizes both the seller's surplus and the social surplus in this case is ternary information structure. 
Formally, we have the following proposition, which is the analog of Proposition~\ref{prop: buyer_seller_socially_optimal_info}.

  \begin{proposition} \label{prop:sticky_price}
$\mathcal{V}_S(\pi|t)$, $\mathcal{V}_B(\pi|t)$, and $\mathcal{V}^*(\pi|t)$ are maximized when $\pi$ is equivalent to the ternary signal $\pi'$ with $\supp{\mu'}\subset\{0,\frac{1}{2},1\}$. 
Let $\varepsilon^*_S$, $\varepsilon^*_B$, and $\varepsilon^*$ be the values of the probability of receiving an uninformative signal that maximize $\mathcal{V}_S(\pi|t)$, $\mathcal{V}_B(\pi|t)$, and $\mathcal{V}^*(\pi|t)$, respectively. 
Then, $\varepsilon^*_S=(\frac{t+1-(t-1)\delta^t-\sqrt{[t+1-(t-1)\delta^t]^2-4\delta^t}}{2\delta^t})^\frac{1}{t}$, $\varepsilon^*_B=0$, and $\varepsilon^*=0$ if $\alpha\geq \frac{1}{2}$. 
\end{proposition}

\begin{proof}
First, 
\begin{align*}
\mathcal{V}_B(\pi|t)&=(1-\delta)\sum_{i=1}^\infty \delta^{i-1} [V(\pi)+\mathcal{V}_i(\pi)-p_i(\pi|t)]\\
&=(1-\delta)\sum_{i=1}^\infty \delta^{i-1} \left[ \max_{\bm{\sigma}\in B^i(\pi)} V_i(\pi,\bm{\sigma})-p_i(\pi|t)\right]\\
&\leq (1-\delta)\sum_{i=1}^\infty \delta^{i-1} \max_{\bm{\sigma}\in B^i(\pi)} V_i(\pi,\bm{\sigma}) \\
&\leq \frac{1}{4},
\end{align*} 
where the last inequality holds strictly if $\pi$ is not full information. 
Let $\overline{\pi}$ be full information. 
Then, $V(\overline{\pi})+\mathcal{V}_i(\overline{\pi})-p_i(\overline{\pi})=\frac{1}{4}+0-0=\frac{1}{4}$ for all $i$. Thus, the information structure that maximizes buyer surplus is full information, or $\varepsilon_B^*=0$ 

The seller's surplus is monotonically increasing with respect to $\mathcal{V}_i(\pi)$ (for $i = 1, t+1, 2t+1, \dots$). 
Thus, Theorem~\ref{thm:ternary_signal} implies that the information structure which maximizes $\mathcal{V}_S(\pi|t)$ is equivalent to $\pi'$ that satisfies $\supp{\mu'}\subset\{0,\frac{1}{2},1\}$.
Then,
\begin{align*}
     \mathcal{V}_{S}(\pi(\varepsilon)|t)
    &=(1-\delta^t) \sum_{k=0}^\infty \delta^{kt}\mathcal{V}_{kt+1}(\pi(\varepsilon))\\
    &=\frac{1-\delta^t}{4} \sum_{k=0}^\infty \delta^{kt}(\varepsilon-\varepsilon^{kt+1})\\
    &=\frac{\delta^t}{4}\cdot\frac{\varepsilon(1-\varepsilon^t)}{1-\delta^t\varepsilon^t}.
\end{align*}
Define $f(\varepsilon)=\frac{\varepsilon(1-\varepsilon^t)}{1-\delta^t\varepsilon^t}$. 
We have
\begin{align*}
    f'(\varepsilon)=\frac{\delta^t\varepsilon^{2t}+[(t-1)\delta^t-(t+1)]\varepsilon^t+1}{(1-\delta^t\varepsilon^t)^2}.
\end{align*}
Note that $f'(0)=1>0$ and $f'(1)=\frac{t(\delta^t-1)}{(1-\delta^t)^2}<0$ hold. 
Since the numerator of $f'(\varepsilon)$ is a quadratic with respect to $\varepsilon^t$, $f(\varepsilon)$ is maximized at $\varepsilon_S^*$ which satisfies $f'(\varepsilon^*_S)=0$. 
The simple calculation yields
 \begin{align*}
     \varepsilon^*_S=
     \left(\frac{t+1-(t-1)\delta^t-\sqrt{[t+1-(t-1)\delta^t]^2-4\delta^t}}{2\delta^t}\right)^\frac{1}{t}.
 \end{align*}
Consider  $\mathcal{V}^*(\pi|t)$. 
Note that
 \begin{align*}
     \mathcal{V}^*(\pi|t)&=\alpha \mathcal{V}_B(\pi|t)+(1-\alpha)\mathcal{V}_S(\pi|t)\\&=
     \alpha[V(\pi)+(1-\delta)\sum_{i=1}^\infty \delta^{i-1} \mathcal{V}_i(\pi)-(1-\delta^t)\sum_{k=0}^\infty \delta^{kt}\mathcal{V}_{kt+1}(\pi)]\\&\quad\quad\quad\quad\quad\quad\quad+(1-\alpha)(1-\delta^t)\sum_{k=0}^\infty \delta^{kt}\mathcal{V}_{kt+1}(\pi)\\
     &=\alpha V(\pi)+\alpha(1-\delta)\sum_{i=1}^\infty \delta^{i-1} \mathcal{V}_i(\pi)+(1-2\alpha)(1-\delta^t)\sum_{k=0}^\infty \delta^{kt}\mathcal{V}_{kt+1}(\pi) 
     \end{align*}
     If $\alpha\leq 1/2$, then all terms are increased by replacing them with the ternary information structure. Hence, the information structure which maximizes $\mathcal{V}^*(\pi|t)$ is equivalent to $\pi'$ that satisfies $\supp{\mu'}\subset\{0,\frac{1}{2},1\}$.
If $\alpha\geq 1/2$, then 
 \begin{align*}
     \mathcal{V}^*(\pi|t)
     &=\alpha V(\pi)+\alpha(1-\delta)\sum_{i=1}^\infty \delta^{i-1} \mathcal{V}_i(\pi)+(1-2\alpha)(1-\delta^t)\sum_{k=0}^\infty \delta^{kt}\mathcal{V}_{kt+1}(\pi) \\
     & \leq \alpha V(\pi)+\alpha(1-\delta)\sum_{i=1}^\infty \delta^{i-1} \mathcal{V}_i(\pi) \\
     & = \alpha(1-\delta)\sum_{i=1}^\infty \delta^{i-1} [\mathcal{V}_i(\pi) + V(\pi)] \\
     &=\alpha(1-\delta)\sum_{i=1}^\infty \delta^{i-1}\max_{\bm{\sigma}\in B^i(\pi)}V_i(\pi,\bm{\sigma})\\
     & \leq \frac{\alpha}{4}. 
\end{align*}
Let $\overline{\pi}$ be full information. 
Then, $\mathcal{V}^*(\overline{\pi}|t)= \frac{\alpha}{4}$. 
Thus, $\varepsilon^*=0$ is optimal if $\alpha \geq \frac{1}{2}$. 
\end{proof}


\titleformat{\section}
		{\Large\bfseries}     
         {Appendix \thesection:}
        {0.5em}
        {}
        []

\renewcommand{\thetheorem}{A.\arabic{theorem}}
\setcounter{theorem}{0}

 \appendix 

\section{Omitted Proofs}
 To prove this theorem, we first prove the following lemma.
\begin{lemma}\label{lem:split}
     If $\supp{\mu_i}\cap(0,\frac{1}{2})\neq \emptyset$ and $\supp{\mu_i}\cap(\frac{1}{2},1)\neq \emptyset$ for $i=1,2$, there exists $x \in \supp{\mu_1\otimes \mu_2}$ such that 
         \begin{align*}
            & 0\in \supp{\mu_1^*\otimes \mu_2^*|(\mu_1\otimes \mu_2)=x}\\
             &1\in \supp{\mu_1^*\otimes \mu_2^*|(\mu_1\otimes \mu_2)=x},
         \end{align*}
         where $\pi^{*}$ is an information structure derived from following split from $\pi$: splitting $l\in\supp{\mu}\cap (0,1/2)$ and $h\in \supp{\mu}\cap (1/2,1)$ into $\{0,1/2\}$ and $\{1/2,1\}$, respectively.
        \end{lemma}
        \begin{proof}
     Take $y\in \supp{\mu_1}$ and $z \in \supp{\mu_2}$. Without loss of generality, suppose that $S_1(y)=\{s_1\}$, $S_2(z)=\{s_2\}, S^*_1(0)=\{s^{*0}_1\},  S^*_1(\frac{1}{2})=\{s^{*\frac{1}{2}}_1\},  S^*_1(1)=\{s^{*1}_1\},  S^*_2(0)=\{s^{*0}_2\},  S^*_2(\frac{1}{2})=\{s^{*\frac{1}{2}}_2\}$, and $ S^*_2(1)=\{s^{*1}_2\}$. Consider the situation where $s_1$ and $s_2$ are received under $\pi_1\otimes \pi_2$.
          Note that 
          \begin{align*}
              &\pi^*_1(s_1^{*\frac{1}{2}}|H,s_1)=\min\left\{1,\frac{\pi_1(s_1|L)}{\pi_1(s_1|H)}\right\}\\
              &\pi^*_1(s_1^{*0}|H,s_1)=0\\
              &\pi^*_1(s_1^{*1}|H,s_1)=\max\left\{0,\frac{\pi_1(s_1|H)-\pi_1(s_1|L)}{\pi_1(s_1|H)}\right\}\\
              &\pi^*_1(s_1^{*\frac{1}{2}}|L,s_1)=\min\left\{1,\frac{\pi_1(s_1|H)}{\pi_1(s_1|L)}\right\}\\
              &\pi^*_1(s_1^{*0}|L,s_1)=\max\left\{0,\frac{\pi_1(s_1|L)-\pi_1(s_1|H)}{\pi_1(s_1|L)}\right\}\\
              &\pi^*_1(s_1^{*1}|L,s_1)=0
          \end{align*}
          and
            \begin{align*}
              &\pi^*_2(s_2^{*\frac{1}{2}}|H,s_2)=\min\left\{1,\frac{\pi_2(s_2|L)}{\pi_2(s_2|H)}\right\}\\
              &\pi^*_2(s_2^{*0}|H,s_2)=0\\
              &\pi^*_2(s_2^{*1}|H,s_2)=\max\left\{0,\frac{\pi_2(s_2|H)-\pi_2(s_2|L)}{\pi_2(s_2|H)}\right\}\\
              &\pi^*_2(s_2^{*\frac{1}{2}}|L,s_2)=\min\left\{1,\frac{\pi_2(s_2|H)}{\pi_2(s_2|L)}\right\}\\
              &\pi^*_2(s_2^{*0}|L,s_2)=\max\left\{0,\frac{\pi_2(s_2|L)-\pi_2(s_2|H)}{\pi_2(s_2|L)}\right\}\\
              &\pi^*_2(s^{*1}|L,s_1)=0.
          \end{align*}
        Therefore,
        \begin{align*}
            &(\pi_1^*\otimes \pi_2^*)((s_1^{*\frac{1}{2}},s_2^{*\frac{1}{2}})|H,(s_1,s_2))=\min\left\{1,\frac{\pi_1(s_1|L)}{\pi_1(s_1|H)}\right\}\min\left\{1,\frac{\pi_2(s_2|L)}{\pi_2(s_2|H)}\right\}\\
            &(\pi_1^*\otimes \pi_2^*)((s_1^{*\frac{1}{2}},s_2^{*\frac{1}{2}})|L,(s_1,s_2))=\min\left\{1,\frac{\pi_1(s_1|H)}{\pi_1(s_1|L)}\right\}\min\left\{1,\frac{\pi_2(s_2|H)}{\pi_2(s_2|L)}\right\}\\
        \end{align*}
        Take $y\in \supp{\mu_1}\cap(0,\frac{1}{2})$ and $z \in \supp{\mu_2}\cap(\frac{1}{2},1)$. Let $x=\frac{yz}{yz+(1-y)(1-z)}$. Then,
        \begin{align*}
         &   (\mu^*_1\otimes \mu^*_2 )(\frac{1}{2}|H,\mu_1=y,\mu_2=z)=\frac{\pi_2(s_2|L)}{\pi_2(s_2|H)}\\
         & (\mu^*_1\otimes \mu^*_2 )(1|H,\mu_1=y,\mu_2=z)=\frac{\pi_2(s_2|H)-\pi_2(s_2|L)}{\pi_2(s_2|H)}\\
           &   (\mu^*_1\otimes \mu^*_2 )(\frac{1}{2}|L,\mu_1=y,\mu_2=z)=\frac{\pi_1(s_1|H)}{\pi_1(s_1|L)}\\
         & (\mu^*_1\otimes \mu^*_2 )(0|L,\mu_1=y,\mu_2=z)=\frac{\pi_1(s_1|L)-\pi_1(s_1|H)}{\pi_1(s_1|L)}\\
        \end{align*}
        Since $Prob(H|\mu_1=y,\mu_2=z)=x$, we have
        \begin{align*}
            &(\mu^*_1\otimes \mu^*_2 )(1|\mu_1=y,\mu_2=z)=x\cdot \frac{\pi_2(s_2|H)-\pi_2(s_2|L)}{\pi_2(s_2|H)}>0\\
            &(\mu^*_1\otimes \mu^*_2 )(0|\mu_1=y,\mu_2=z)=(1-x)\cdot \frac{\pi_1(s_1|L)-\pi_1(s_1|H)}{\pi_1(s_1|L)}>0
        \end{align*}
        Therefore, 
        \begin{align*}
            & 0\in \supp{\mu_1^*\otimes \mu_2^*|(\mu_1\otimes \mu_2)=x}\\
             &1\in \supp{\mu_1^*\otimes \mu_2^*|(\mu_1\otimes \mu_2)=x}.
         \end{align*}
\end{proof}
\begin{lemma}\label{lem:equivalent}
    If $\pi$ and $\pi'$ satisfy $\supp{\mu}\cap(0,1)\subset (0,\frac{1}{2}]$, $\supp{\mu'}\cap(0,1)\subset (0,\frac{1}{2}]$, and $\pi(\mu=1|H)=\pi'(\mu'=1|H)$, then $\pi\sim \pi'$. Similarly, if $\pi$ and $\pi'$ satisfy $\supp{\mu}\cap(0,1)\subset [\frac{1}{2},1)$, $\supp{\mu'}\cap(0,1)\subset [\frac{1}{2},1)$, and $\pi(\mu=0|L)=\pi'(\mu'=0|L)$, then $\pi\sim \pi'$.
\end{lemma}
\begin{proof}
Suppose that  $\pi$ and $\pi'$ satisfy $\supp{\mu}\cap(0,1)\subset (0,\frac{1}{2}]$, $\supp{\mu'}\cap(0,1)\subset (0,\frac{1}{2}]$, and $\pi(\mu=1|H)=\pi'(\mu'=1|H)$.
    First, show that $\overline{V}_i(\pi)=\overline{V}_i(\pi')$ for all $i$. 
    If agent $i$ can observe the signal $i$ times, it is one of the optimal strategies for agent $i$ to take action 1 if and only if he or she receives a signal that induces $\mu=1$ at least once. 
    Hence,
    \begin{align*}
        \overline{V}_i(\pi)&=\frac{1}{4}\cdot [1-(1-\pi(\mu=1|H))^i]\\
        &=\frac{1}{4}\cdot[1-(1-\pi'(\mu'=1|H))^i]\\
        &=\overline{V}_i(\pi').
    \end{align*}
    Especially, we have $V(\pi)=V(\pi')$.

    Next, show that $\mathcal{V}_i(\pi)=\overline{V}_i(\pi)-V(\pi)$.
    By Lemma 5 in \cite{sato2025value}, we know that $\mathcal{V}_i(\pi)\leq \overline{V}_i(\pi)-V(\pi)$.
    Consider the following strategy profile $\bm{\sigma}$: agent $i$ takes action 1 if and only if he or she receives a signal that induces $\mu=1$ or at least one agent before $i$ takes action 1. Then, this is an equilibrium and it follows that
    $V_i(\pi,\bm{\sigma})=\overline{V}_i(\pi)$. 
    Hence, $\mathcal{V}_i(\pi)\geq \overline{V}_i(\pi)-V(\pi)$. 
    Thus, $\mathcal{V}_i(\pi)= \overline{V}_i(\pi)-V(\pi)$. 
    By the same argument, we have $\mathcal{V}_i(\pi')= \overline{V}_i(\pi')-V(\pi')$. 
    Therefore, $\mathcal{V}_i(\pi)=\mathcal{V}_i(\pi')$ for all $i$, and  we have $\pi\sim \pi'$.
    Since the other case is symmetric, the proof is omitted.
\end{proof}
\begin{proof}[Proof of Theorem 1]
    Take any information structure $\pi$.
    Let $\pi^{*}$ be the information structure derived from the split mentioned in Lemma~\ref{lem:split}.
    Then, $\pi^{*} \succsim_{B} \pi$. 
    Note that
     \begin{align*}
        &V(\pi)=\sum_{x\in \supp{\mu}\cap(\frac{1}{2},1]}\left[\frac{x}{2}-\frac{1-x}{2}\right]\cdot \pi(\mu=x)\\
        &\quad \quad\,\,=\sum_{x\in \supp{\mu}\cap(\frac{1}{2},1]} \left[x-\frac{1}{2}\right]\cdot \pi(\mu=x)\\
        &\quad\quad\,\,=\pi(\mu> \frac{1}{2})\cdot \left(\mathbb{E}_\pi\left[\mu\mid \mu> \frac{1}{2}\right]-\frac{1}{2}\right).\\
        &V(\pi^*)=\frac{1}{2}\cdot \pi^*(\mu^*=1)\\
        &\quad\quad\,\,=\pi^*(\mu^*=1)\cdot \left[1-\frac{1}{2}\right]\\
        &\quad\quad\,\,=\pi^*(\mu^*>\frac{1}{2})\cdot \left(\mathbb{E}_{\pi^*}\left[\mu^*\mid \mu^*> \frac{1}{2}\right]-\frac{1}{2}\right)\\
        &\quad\quad\,\,=\pi(\mu>\frac{1}{2})\cdot \left(\mathbb{E}_{\pi}\left[\mu \mid \mu> \frac{1}{2}\right]-\frac{1}{2}\right).
    \end{align*}
    Hence, $V(\pi)=V(\pi^*)$.
    By Lemma 6 in \citet{sato2025value}, we know that, for all BNE $\sigma^*$,
    \begin{align*}
        V_i(\pi^*,\sigma^*)=\overline{V}_i(\pi^*),
 \end{align*}
    
    Since  $\overline{V}_i(\pi)\geq V_i(\pi,\bm{\sigma})$ (Lemma 5 in  \citet{sato2025value}), it is sufficient to show  that $\overline{V}_i(\pi^*)\geq \overline{V}_i(\pi)$ holds.
    Note that
    \begin{align*}
        \overline{V}_i(\pi)=\sum_{x\in \supp{\mu\otimes\mu} } g(x)\mu^{\otimes i}(x),
        \end{align*}
        where 
        \begin{align*}
            g(x)=\begin{cases}
                0&(x\in[0,\frac{1}{2}])\\
                x-\frac{1}{2}&(x\in [\frac{1}{2},1]).
            \end{cases}
        \end{align*}
        By Lemma 4 in \citet{sato2025value}, $\mu^{*\otimes i} $ is a mean preserving spread of $\mu^{\otimes i}$.
        Thus, \footnote{Precisely, this is not correct because $\mu^{*\otimes i}$ and $\mu^{\otimes i}$ may be independent. Formally, $\mu^{*\otimes i} $ is a mean preserving spread of $\mu^{\otimes i}$ if there exists distribution $F$ and $F^*$ such that
        $\mu^{\otimes i}\sim F$, $\mu^{*\otimes i}\sim F^*$, and $\sum_{y\in \supp{F^* }}y\cdot f^*(y|f=x)=x$. Here, for simplicity, we assume that $\mu^{\otimes i}$ and $\mu^{*\otimes i}$ itself satisfy this equation.}
        
         \begin{align*}
        \sum_{y\in \supp{\mu^{*\otimes i}} }y\cdot \mu^{*\otimes i}(y|\mu^{\otimes i}=x)=x,
        \end{align*}
        for each $x$. 
        Then,
        \[
    \sum_{y\in \supp{\mu^{*\otimes i}} }g(y)\cdot\mu^{*\otimes i}(y|\mu^{\otimes i}=x)
    \geq  g\left(\sum_{y\in \supp{\mu^{*\otimes i}} }y\cdot\mu^{*\otimes i}(y|\mu^{\otimes i}=x)\right)
    = g(x),
        \]
        where the inequality comes from Jensen's inequality.

Since
            \begin{align*}
        \overline{V}_i(\pi^*)=\sum_{x\in \supp{\mu^{\otimes i}} }\left[\sum_{y\in \supp{\mu^{*\otimes i}} } g(y)\cdot \mu^{*\otimes i}(y|\mu^{\otimes i}=x)\right]\cdot \mu^{\otimes i}(x),
        \end{align*}
         we have 
         \begin{align*}
             \overline{V}_i(\pi^*)&\geq\sum_{x\in \supp{\mu^{\otimes i}}}g(x)\cdot \mu^{\otimes i}(x)\\&= \overline{V}_i(\pi)
         \end{align*} for all $i$. Also, we have
         $\overline{V}_i(\pi^*)>\overline{V}_i(\pi)$ if there exists $x \in \supp{\mu^{\otimes i}}$ such that
         \begin{align*}
             \sum_{y\in \supp{\mu^{*\otimes i}} }g(y)\cdot \mu^{*\otimes i}(y|\mu^{\otimes i}=x)>g(x).
         \end{align*}
         In other words, $\overline{V}_i(\pi^*)>V_i(\pi,\sigma)$ if there exists $x \in \supp{\mu^{\otimes i}}$ such that 
         \begin{align*}
            & 0\in \supp{\mu^{*\otimes i}|\mu^{\otimes i}=x}\\
             &1\in  \supp{\mu^{*\otimes i}|\mu^{\otimes i}=x}.
         \end{align*}
    
By Lemma~\ref{lem:split}, if $\supp{\mu}\cap(0,\frac{1}{2})\neq \emptyset$ and $\supp{\mu}\cap(\frac{1}{2},1)\neq \emptyset$, we have $\overline{V}_2(\pi^*)>\overline{V}_2(\pi)$ and 
         $\supp{\mu\otimes \mu}\cap(0,\frac{1}{2})\neq \emptyset$ and $\supp{\mu\otimes \mu}\cap(\frac{1}{2},1)\neq \emptyset$.
         Hence, $\overline{V}_3(\pi^*)>\overline{V}_3(\pi)$. Analogously, it follows that
         $\overline{V}_i(\pi^*)>\overline{V}_i(\pi)$ for all $i\geq 2$ if $\supp{\mu}\cap(0,\frac{1}{2})\neq \emptyset$ and $\supp{\mu}\cap(\frac{1}{2},1)\neq \emptyset$.
         Hence, the value of history for each agent $i$ ($i\geq 2$) is maximized at $\pi$ only if $\supp{\mu}\cap(0,1)\subset (0,\frac{1}{2}]$ or $\supp{\mu}\cap(0,1)\subset [\frac{1}{2},1)$. Suppose that $\pi$ satisfies this condition. Then,
         Lemma~\ref{lem:equivalent} guarantees that there exists $\pi'$ such that $\pi\sim \pi'$ and $\supp{\mu'}\subset\{0,\frac{1}{2},1\}$. Since $\mathcal{V}_i(\pi)=0$ when $\pi$ is either no or full information, Theorem 1 holds.
    
\end{proof}

\singlespacing 
\bibliography{Reference}
\addcontentsline{toc}{section}{Reference}

\end{document}